\documentclass[sigconf]{acmart}

\usepackage[utf8]{inputenc}

\usepackage{amsmath,amssymb}
\usepackage{amsthm}
\usepackage{mathrsfs}
\usepackage{stmaryrd}
\usepackage{enumerate}
\usepackage[algoruled,vlined,english,linesnumbered]{algorithm2e}
\usepackage{comment}
\usepackage{multirow}
\usepackage{xspace}

\usepackage[inline]{enumitem}

\newcommand{\noopsort}[1]{}

\definecolor{purple}{rgb}{0.6,0,0.6}

\definecolor{answer}{rgb}{0,0.5,0.2}

\newtheorem{theo}{Theorem}[section]
\newtheorem{lem}[theo]{Lemma}
\newtheorem{prop}[theo]{Proposition}
\newtheorem{cor}[theo]{Corollary}

\theoremstyle{definition}
\newtheorem{rem}[theo]{Remark}

\newtheorem{deftn}[theo]{Definition}

\copyrightyear{2018}
\acmYear{2018}
\setcopyright{acmlicensed}
\acmConference[ISSAC '18]{2018 ACM International Symposium on Symbolic and Algebraic Computation}{July 16--19, 2018}{New York, NY, USA}
\acmBooktitle{ISSAC '18: 2018 ACM International Symposium on Symbolic and Algebraic Computation, July 16--19, 2018, New York, NY, USA}
\acmPrice{15.00}
\acmDOI{10.1145/3208976.3209012}
\acmISBN{978-1-4503-5550-6/18/07}

\begin{document}

\title{On Affine Tropical F5 Algorithms}

\author{
Tristan Vaccon}
  \affiliation{Universit\'e de Limoges;
  \institution{CNRS, XLIM UMR 7252}
  \city{Limoges, France}  
  \postcode{87060}  
  }
  \email{tristan.vaccon@unilim.fr}

\author{
Thibaut Verron}
  \affiliation{
  \institution{Johannes Kepler University\\Institute for Algebra }
  \city{Linz, Austria}  
}
\email{thibaut.verron@jku.at}

\thanks{The second author is supported by the Austrian FWF grant F5004.}

\author{Kazuhiro Yokoyama}
  \affiliation{Departement of Mathematics,
  \institution{Rikkyo University}
  \city{Tokyo, Japan}}
  \email{kazuhiro@rikkyo.ac.jp}

\begin{abstract}
Let $K$ be a field equipped with a valuation. Tropical varieties over $K$ can be defined with a theory of Gröbner bases taking into account the valuation of $K$.
Because of the use of the valuation, the theory of tropical Gröbner bases 
has proved to provide settings for 
computations over polynomial rings over a $p$-adic field
that are more stable than that of classical Gröbner bases.

Beforehand, these strategies were only available for homogeneous
polynomials. 
In this article, we extend the F5 strategy to a new definition of 
tropical Gröbner bases in an affine setting.

We provide numerical examples to illustrate
time-complexity and $p$-adic stability of this tropical F5 algorithm.
We also illustrate its merits as a first step before
an FGLM algorithm to compute (classical) lex bases over $p$-adics.

\end{abstract}



\vspace{-1.5mm}
\terms{Algorithms, Theory}

\maketitle

\section{Introduction}

Tropical geometry as we understand it has not yet reached half a century of age.
It has nevertheless spawned significant applications to very various domains, from algebraic geometry to combinatorics, computer science, economics, non-archimedean geometry (see \cite{MS:2015}, \cite{EKL06}) and even attempts at proving the Riemann hypothesis (see \cite{Connes:2015b}).

Effective computation over tropical varieties make decisive use of Gröbner bases.
Since Chan and Maclagan's definition of tropical Gröbner bases taking into account the valuation in \cite{Chan:2013,CM:2013}, computations of 
tropical Gröbner bases are available over fields with trivial or non-trivial valuation, but only in a context of homogeneous ideals.

On the other hand, for classical Gröbner bases, numerous algorithms have been developed allowing for more and more efficient computations.
The latest generation of algorithms for computing Gröbner bases is the family of signature-based algorithms, which keep track of where the polynomials come from in order to anticipate useless reductions.
This idea was initiated in Algorithm F5 \cite{F5}, and has since then been widely studied and generalized (\cite{BFS14,EF17}).

Most of those algorithms, including the original F5 algorithm, are specifically designed for homogeneous systems, and adapting them to affine (or inhomogeneous) systems requires special care (see \cite{E13}).

An F5 algorithm computing tropical Gröbner bases without any trivial reduction to 0, inspired by the classical F5 algorithm, has been described in \cite{Vaccon:2017}.
The goal of this paper is to extend the definition of tropical Gröbner bases to inhomogeneous ideals, and describe ways to adapt the F5 algorithm in this new setting.

The core motivation is the following.
It has been proved~\cite{Vaccon:2015} that computing tropical Gröbner bases, taking into account the valuation, is more stable for polynomial ideals over a $p$-adic field than classical Gröbner bases.

Thus, an affine variant of tropical Gröbner bases is highly desirable to handle non-homogeneous ideals over $p$-adics.
For classical Gröbner bases, it is always possible to homogenize the input ideal, compute a homogeneous Gröbner basis, and dehomogenize the result.
This technique is not always optimal, because the algorithm may end up reaching a higher degree than needed, computing points at infinity of the system, but it always gives a correct result and, in the case of signature Gröbner basis algorithms, is able to eliminate useless reductions.
However, in a tropical setting, terms are ordered with a tropical term order, taking into account the valuation of the coefficients.
As far as we know it, there is no way to dehomogenize a system
 in a 
way that would preserve the tropical term order. 
Indeed, no tropical term order can be an elimination order.

Moreover, the FGLM algorithm can be adapted to the 
tropical case (see Chap. 9 of \cite{Vaccon:these}), making it possible to compute a lexicographical (classical) Gröbner basis from a tropical one.
We provide numerical data to estimate the loss in precision for the computation of a lex Gröbner basis using a tropical F5 algorithm followed by an FGLM algorithm, in an affine setting.

\vspace{-.2cm}

\subsection{Related works} 
A canonical reference for an introduction to computational 
tropical algebraic geometry is the book of Maclagan and Sturmfels 
\cite{MS:2015}.

The computation of tropical varieties over $\mathbb{Q}$ with trivial valuation is available in the Gfan package by Anders Jensen (see \cite{Jensen}), by using standard Gröbner bases computations.
Chan and Maclagan have developed in \cite{CM:2013} a Buchberger algorithm to compute tropical Gröbner bases for homogeneous entry polynomials (using a special division algorithm).
Following their work, still for homogeneous polynomials, a Matrix-F5 algorithm has been proposed in \cite{Vaccon:2015} and a Tropical F5 algorithm in \cite{Vaccon:2017}.
Markwig and Ren have provided a completely different technique of computation using projection of standard bases in \cite{MY:2015}, again only for homogeneous entry polynomials.

In the classical Gröbner basis setting, many techniques have been studied to make the computation of Gröbner bases more efficient.
In particular, Buchberber's algorithm is frequently made more efficient by using the \emph{sugar-degree} (see \cite{GMNRT91, BCM11}) instead of the actual degree for selecting the next pair to reduce.
This technique was a precursor of modern signature techniques, in the sense that the sugar-degree of a polynomial is exactly the degree of its signature.
General signature-based algorithms for computing classical Gröbner bases of inhomogeneous ideals have been extensively studied in~\cite{E13}.

\subsection{Specificities of computating tropical GB}

%

The computation of tropical GB, even by a Buchberger-style algorithm,
is not as straightforward as for classical Gröbner bases.
One way to understand this is the following: even for
homogeneous ideals, there is no equivalence between
tropical Gröbner bases and 
row-echelon linear bases at every degree.
Indeed,we can remark that
 $(f_1,f_2)=(x+y,2x+y)$ is a tropical 
 GB over $\mathbb{Q}[x,y]$ with $2$-adic valuation,
$w=[0,0]$ and grevlex ordering. 
Nevertheless, 
the corresponding $2 \times 2$ matrix in the vector space
of homogeneous polynomials of degree $2$ is not
 under tropical row-echelon form.  

As a consequence, reduction of a polynomial by
a tropical GB is not easy.
In \cite{Chan:2013,CM:2013}, Chan and Maclagan
relied on a variant of Mora's tangent cone algorithm
to obtain a division algorithm.
In \cite{Vaccon:2015, Vaccon:2017}, the authors relied 
on linear algebra and the computation of 
(tropical) row-echelon form.
In this article, we extend their method
to the computation of tropical Gröbner bases
in an affine setting, through an F5 algorithm.

\subsection{Main idea and results}

Extending the tropical F5 algorithm to inhomogeneous inputs poses two difficulties.
First, as mentioned, tropical Gröbner bases used to
be only defined  and computed for homogeneous systems.
Even barebones algorithms such as Buchberger's algorithm are not available for inhomogeneous systems.
The second problem is a general problem of signature Gröbner bases with inhomogeneous input.
The idea of signature algorithms is to compute polynomials with increasing signatures, and the F5 criterion detects trivial reductions to 0 by matching candidate signatures with existing leading terms.
For homogeneous ideals, the degree of the signature of a polynomial and the degree of the polynomial itself are correlated.
This is what makes the F5 criterion applicable.

The survey paper~\cite{E13} has shown that Algorithm F5, using the \emph{position over term}
ordering on the signatures, has to reach a tradeoff between eliminating all reductions to 0 and performing other useless reductions.

More precisely, let $f_{1},\dots,f_{m}$ be homogeneous polynomials with coefficients in a field with valuation $K$, and define $I_{k,d}$ the vector space of polynomials in $\langle f_{1},\dots,f_{k} \rangle$ with degree at most $d$.
With the usual computational strategy, the algorithm computes a basis of $I_{1,1}$, then $I_{2,1}$, and so on until $I_{m,1}$, and then $I_{1,2}$, and so on.
In a lot of situations \cite{BFS2004} ideals with more generators have a Gröbner basis with lower degree, and this strategy ensures that the algorithm does not reach a degree higher than needed.

However, the same algorithm for affine system will, at each step, merely compute a \emph{set of polynomials} in each $I_{k,d}$.
This set needs not be a generating set because of degree falls.
To obtain a basis instead, one has to proceed up to some $I_{k,\delta}$ with $\delta \geq d$.
When $\delta > d$, some polynomials will be missing for the F5 criterion in degree less than $\delta$, and the corresponding trivial reductions to 0 will not be eliminated.


In this paper, we show that the tropical F5 algorithm \cite{Vaccon:2017} works in an affine setting, and we characterize those trivial reductions to 0 which are eliminated by the F5 criterion.
In particular, we show that the Macaulay matrices built at each step of the computations are Macaulay matrices of all polynomials with a given \emph{sugar-degree}.

Compared to \cite{Vaccon:2017}, the overall presentation of the F5 algorithms is clarified. It can now be summarized as the following strategy: filtration, signature, F5 elimination criterion, Buchberger-F5 criterion and finally the F5 algorithm.

\begin{theo} \label{thm_intro}
  Given a set of (non-necessarily homogeneous) polynomials $f_{1},\dots,f_{m} \in K[X_{1},\dots,X_{n}]$, the Tropical F5 algorithm (Algorithm \ref{F5_algo_complete})
  computes a tropical Gröbner basis of $I$, without reducing to 0 any trivial tame syzygy (Def.~\ref{def:tame-syz}).
\end{theo}

We also examine an incremental affine version of the homogeneous tropical F5-algorithm and an affine tropical F4, and we compare their performances on several examples.
Even in a non-homogeneous setting, the loss in precision of the
tropical F5 algorithm remains satisfyingly low.

\subsection{Organization of the paper}
Section \ref{sec:Aff_Trop_GB} introduces notations and nonhomogeneous tropical Gröbner bases.
Section \ref{sec:Filtration} then introduces the filtration on ideals necessary for F5 algorithms in this context.
Section \ref{sec:BF5_criterion} is devoted to provide a Buchberger-F5 criterion on which Section \ref{sec:F5_algo} elaborates a first tropical F5 algorithm.
Section \ref{sec:Other_algo} briefly presents other methods for the computation of nonhomogeneous tropical Gröbner bases.
Finally, Section \ref{sec:Num_exp} displays numerical results related to the precision behaviour and time-complexity of the algorithms we have described.

\subsection*{Acknowledgements}
We thank Jean-Charles Faugère, Marc Mezzarobba, Pierre-Jean Spaenlehauer, Masayuki Noro,
and Naoyuki Shinohara for fruitful discussions.

\section{Affine Tropical GB}
\label{sec:Aff_Trop_GB}
%

\subsection{Notations} \label{subsec:Notations}
Let $k$ be a field with valuation  $val.$
The polynomial ring $k[X_1,\dots, X_n]$ will be denoted by $A.$ Let $T$ be the set of monomials of $A.$
For $u=(u_1,\dots,u_n) \in \mathbb{Z}_{\geq 0}^n$, we write $x^u$ for
$X_1^{u_1} \dots X_n^{u_n}$ and $\vert f \vert$ for the degree of $f \in A.$
In $A^s,$ $(e_i)_{i=1}^s$ is the canonical basis.

The matrix of a list of polynomials written in a basis of monomials
is called a \textit{Macaulay matrix}.

Given a mapping $\phi : U \rightarrow V,$ $Im(\phi)$
denotes the image of $\phi.$ For a matrix $M,$ $Rows(M)$ is the list
of its rows, and $Im(M)$ denotes
the left-image of $M$ (\textit{i.e.} $Im(M)=span(Rows(M)$).
For $w \in Im(val)^n \subset \mathbb{R}^n$ and $\leq_1$ a monomial order on $A,$
we define $\leq$ a tropical term order as in the following definition:

\begin{deftn} \label{defn:trop_term_order}
Given $a,b \in k^*$ and $x^\alpha$ and $x^\beta$ two monomials in $A$, 
we write $a x^\alpha < b x^\beta$ if:
\begin{itemize}
\item $\vert x^\alpha \vert < \vert x^\beta \vert,$ or
\item $\vert x^\alpha \vert = \vert x^\beta \vert,$
and
$val(a)+w \cdot \alpha > val(b) +w \cdot \beta$, or
$val(a)+w \cdot \alpha = val(b) +w \cdot \beta$ 
and $x^\alpha <_1 x^\beta.$
\end{itemize} 
For $u$ of valuation $0,$ we write $a x^\alpha =_{\leq} u a x^\alpha.$
Accordingly, $a x^\alpha \leq b x^\beta$
if $a x^\alpha < b x^\beta$ or $a x^\alpha =_{\leq} b x^\beta.$
\end{deftn}

Throughout this article, we are
interested in computing a tropical Gröbner basis
of $I=\left\langle f_1, \dots, f_s \right\rangle$
for some given $f_1,\dots, f_s \in A$ (ordered
increasingly by degree).

\subsection{Tropical GB}

A tropical term order provides an order on the terms of the polynomials $f \in A.$

\begin{deftn}
For $f \in A,$ we define $LT(f)$ to be the biggest 
term of $f.$ We define $LM(f)$ to be
the monomial corresponding to $LT(f)$ and
$LC(f)$ the corresponding coefficient.

We define $LM(I)$ to be the monomial ideal
generated by the monomials $LM(f)$ for $f \in I.$
\end{deftn}

We can then naturally define what is a tropical 
Gröbner basis (\textit{tropical GB} for short):

\begin{deftn}
$G \subset I$ is a tropical GB of $I$
if $\text{span}(LM(g) \text{ for } g \in G)= LM(I).$
\end{deftn}

We can remark that for homogeneous polynomials this definition coincide with that
given  in \cite{Vaccon:2017}.

\section{Filtration and $\mathfrak{S}$-GB}
\label{sec:Filtration}

\subsection{Definition and elimination criterion}

One of the main ingredient for F5 algorithms
is the definition of a vector space filtration of the ideal $I.$
It is defined from the initial polynomials $F=(f_1,\dots,f_s)$
generating $I.$ 
For simplicity, we assume that they are ordered
by increasing degree.

First, we extend $\leq$  
to the monomials of the vector space $A^s$.
To that intent, we highlight some monomials that appear as
leading monomial of a syzygy.

\begin{deftn}
  \label{def:tame-syz}
Let $(a_1,\dots,a_s) \in A^s$ and $i \in \{1,\dots,s\}$ be such that:
\begin{enumerate*}
\item $\sum_j a_j f_j=0.$
\item $a_i \neq 0$ and $a_j =0$ for $j>i.$
\item for all $j<i,$ $\vert a_j f_j \vert \leq \vert a_i f_i \vert.$
\end{enumerate*}

We call such a syzygy a \textit{tame syzygy} and
we define $LM(a_i)e_i$ to be its leading monomial.
We define $LM(TSyz(F))$ as the module
in $A^s$
generated by the leading monomials of the tame syzygies.
Trivial tame syzygies are the tame syzygies that are also trivial (\textit{i.e.} in
the module generated by the $f_i e_j-f_j e_i$).
\end{deftn}

The F5 criterion that we use in this article is designed
to recognize some of the tame syzygies
and use this knowledge to avoid useless reduction
to zero of some polynomials. It is the main
motivation for defining a filtration on the vector space $A^s.$
We use a degree-refining monomial ordering $\leq_m$ on $A.$\footnote{$\leq_m$ is not necessarily related to $\leq_1$ and $\leq$.} We define a total order on the monomials of $A^s.$

\begin{deftn}
We write that $x^\alpha e_i \leq_{sign} x^\beta e_j$ if:
\begin{enumerate}
\item if $i < j$, or
\item if $i = j$  and 
 $\vert x^\alpha f_i \vert < \vert x^\beta f_j \vert$, or
\item if $i = j$  and 
 $\vert x^\alpha f_i \vert = \vert x^\beta f_j \vert,$
and 
\begin{itemize}
\item $x^\alpha \notin LM(TSyz(F))$
and $x^\beta \in LM(TSyz(F))$, or
\item both $x^\alpha, x^\beta \in LM(TSyz(F))$ and $x^\alpha \leq_m x^\beta$, or
\item both  $x^\alpha, x^\beta \notin LM(TSyz(F))$ and $x^\alpha \leq_m x^\beta$. 
\end{itemize}
\end{enumerate} \label{deftn:order_for_the_filtration}
\end{deftn}

\begin{deftn} \label{def:ideal_filtration}
  We consider the vector space
  \[ I_{ \leq_{sign} x^\alpha e_i}:= Span ( \{ x^\beta f_j, \text{ s.t. } x^\beta e_j \leq_{sign} x^\alpha e_i \} ) \]
and the vector space $I_{<_{sign} x^\alpha e_i}$ defined accordingly.
We define 
$I= \bigcup_{\uparrow x^\alpha e_i } I_{\leq_{sign} x^\alpha e_i}$
as an increasing vector space filtration of $I.$
\end{deftn}
We then have a very natural definition of signature.
In litterature, this notion of signature 
is sometimes called \textit{minimal signature}.
\begin{deftn} \label{def:signature}
For $f \in I,$ the smallest $x^\alpha e_i$ such that
$f \in I_{\leq_{sign} x^\alpha e_i}$ is called the \textbf{signature}
of $f$ and noted $S(f).$
\end{deftn}
The degree $\vert x^\alpha f_i \vert$ is called the 
\textbf{sugar-degree} 
of $x^\alpha e_i.$
\footnote{Sugar-degree has been introduced and explored in \cite{GMNRT91, BCM11}.}
For the purpose of Algorithm \ref{F5_algo_complete},
we design a filtration compatible with the sugar-degree.

\begin{deftn} \label{def:sugar_ideal_filtration}
  We consider the vector space
  \[ I^{\leq  d}= Span ( \{ x^\beta f_j, \text{ s.t. } \vert x^\beta e_j \vert \leq d \} ) \]
We then define, for $x^\alpha e_i$ with sugar-degree $d$,
the vector space \\$I^{\leq  d}_{ \leq_{sign} x^\alpha e_i}= Span ( \{ x^\beta f_j, \text{ s.t. } x^\beta e_j \leq_{sign} x^\alpha e_i \text{ and } \vert x^\beta f_j \vert \leq d \} )$.
\end{deftn}
$I=\bigcup_{\uparrow d } I^{\leq  d}$ is also a vector space filtration.
$I^{\leq  d}$ can itself be filtrated by the $I^{\leq  d}_{ \leq_{sign} x^\alpha e_i}.$
We have a compatible notion of signature:
\begin{deftn} \label{def:signature_sugar}
For $d\in \mathbb{Z}_{>0}$ and $f \in I^{\leq d},$ the smallest $x^\alpha e_i$ such that
$f \in I^{\leq d}_{\leq_{sign} x^\alpha e_i}$ is called the $d$-\textbf{signature}
of $f$ and noted $S_d(f).$
\end{deftn}
We remark that $S_d(f)$ can be different from $S(f)$
for small $f$, but all $S_d(f)$ are equal when $d$
is large.

The main motivation for defining the vector spaces 
$I^{\leq  d}_{ \leq_{sign} x^\alpha e_i}$
is their finite dimension.
Their compatibility with the sugar-degree allows the F5 algorithm to compute only one Macaulay matrix by sugar-degree $d.$

The goal of the F5 criterion is to recognize
some $x^\alpha e_i$ such that the filtration is constant
at $I^{\leq d}_{\leq_{sign} x^\alpha e_i}$.
As a consequence, this knowledge allows to skip some calculation
as, because of this constancy, they will not provide any
new leading monomial.
We can then state a first version of the F5 elimination criterion:

\begin{prop}[\cite{F5}]
If $x^\alpha$ is such that $x^\alpha e_i \in LM(Tsyz(F)),$
$d \geq \vert x^\alpha f_i \vert,$
then 
the filtration is constant
at $I^{\leq d}_{\leq_{sign}  x^\alpha e_i}. $ \label{prop:F5elimination}
If $x^\alpha \in LM(I^{\leq d}_{\leq_{sign} x^\beta e_j})$ for some
$x^\beta$ and $j$ such that $\vert x^\beta f_j \vert \leq \vert x^\alpha\vert$,
then $x^\alpha e_i \in LM(Tsyz(F))$ for any $i>j.$ 
\end{prop} 
\begin{proof}
For the first part, we can write $(x^\alpha + g)f_i = \sum_{j<i} a_j f_j,$
with $LT(g)<x^\alpha$ and for all $j<i,$ $\vert a_j f_j \vert \leq \vert x^\alpha f_i \vert$.
 Then:
\[x^\alpha f_i = (-g)f_i + \sum_{j=1}^{i-1} a_j f_j. \]

By linear algebra and a complete reduction using as pivot the
$x^\beta e_j \in LM(Tsyz(F)),$
we can assume that $g$ has no monomial in $LM(TSyz(F))$
and obtain: 
$x^\alpha f_i \in I^{\leq d}_{< x^\alpha e_i}, $
and therefore, the filtration is constant at $I^{\leq d}_{\leq  x^\alpha e_i}.$

For the second part, we can write $x^\alpha + g = \sum_{k \leq j} a_k f_k,$
with $LT(g)<x^\alpha$ and for all 
$k \leq j,$ $\vert a_j f_j \vert \leq \vert x^\beta f_j \vert \leq \vert x^\alpha \vert $.
Then $(x^\alpha+g)f_i-\sum_{k \leq j} (a_kf_i) f_k=0$
and we do have $\vert x^\alpha f_i \vert \geq \vert (a_kf_i) f_k \vert$ for 
all $k \leq j.$
\end{proof}

If all the $f_i$'s are homogeneous, this coincides with the usual F5 elimination
criterion, as for example stated in \cite{Vaccon:2017}, which eliminates all trivial reductions to zero in the course of the algorithm.
For affine polynomials, the F5 criterion only eliminates those trivial reductions which are tame.

\subsection{Tropical $\mathfrak{S}$-GB}

In order to take advantage of the F5 elimination criterion 
to compute tropical Gröbner bases, we focus on the 
computation of tropical Gröbner bases which are
compatible with the filtration: tropical $\mathfrak{S}$-GB.
We first need the definition of reductions compatible
with the filtration and the corresponding irreducible polynomials.

\begin{deftn}[$\mathfrak{S}$-reduction]
  
Let $e,g \in I,$ $h \in I.$
We say that $e$ 
$\mathfrak{S}$\textbf{-reduces} 
to $g$ with $h,$
$e \rightarrow^h_{\mathfrak{S}} g, $
if there are $t \in T$ and $\alpha \in k^*$ such that:
\begin{itemize}
\item $LT(g)<LT(e),$ $LM(g) \neq LM(e)$ and $e-\alpha th=g$ and
\item $S(th)<_{sign} S(e).$
\end{itemize} \label{def:S_red}
\end{deftn}


It is then natural to define what is an $\mathfrak{S}$-irreducible polynomial.

\begin{deftn}[$\mathfrak{S}$-irreducible polynomial]
We say that $g \in I$ is $\mathfrak{S}$-irreducible 
if there is no $h \in I$ which $\mathfrak{S}$-reduces $g$.
If there is no
ambiguity, we might omit the $\mathfrak{S}-.$ \label{def:S_irred}
\end{deftn}


\begin{deftn}[Tropical $\mathfrak{S}$-Gröbner basis]
We say that $G \subset I$, a set of $\mathfrak{S}$-irreducible polynomials, is a \textbf{tropical} $\mathfrak{S}$\textbf{-Gröbner basis} (or tropical $\mathfrak{S}-$GB, or just $\mathfrak{S}-$GB for short when there is no ambiguity) of $I$ with respect to a given tropical term order, if 
for each $\mathfrak{S}$-irreducible polynomial $h \in I,$
there exist $g \in G$ and $t \in T$ such that $LM(tg)=LM(h)$
and $tS(g)=S(h).$ \label{def:S-GB}
\end{deftn}

\begin{deftn}
Definitions \ref{def:S_red}, \ref{def:S_irred} and \ref{def:S-GB} have natural
analogues when one restricts to the vector space
$I^{\leq d}$ and $S_d$ with $\mathfrak{S}_d$-reduction, $\mathfrak{S}_d$-irreducible polynomial and tropical $\mathfrak{S}_d$-GB.
\end{deftn}

\begin{prop} \label{prop:reduction_by_S_basis}
If $G$ is a tropical $\mathfrak{S}$-Gröbner basis, then for any nonzero $h \in I,$ there exist $g \in G$ and $t \in T$ such that:
\begin{itemize}
\item $LM(tg)=LM(h)$
\item $S(tg)=tS(g)=S(h)$ if $h$ is irreducible, and $S(tg)=tS(g)<_{sign}S(h)$ otherwise.
\end{itemize}
Hence, there is an $\mathfrak{S}$-reductor for $h$ in $G$ when $h$ is not irreducible.
\end{prop}
\begin{cor}
If $G$ is a tropical $\mathfrak{S}$-Gröbner basis, 
then $G$ is a tropical Gröbner basis of $I,$ for $<.$
\end{cor}
As a consequence computing a tropical $\mathfrak{S}$-GB provides
a tropical GB, and we can use the F5 elimination criterion  \ref{prop:F5elimination} to our advantage when computing
these tropical $\mathfrak{S}$-GB.
Moreover, we also have the following finiteness result:

\begin{prop} \label{prop:finiteness_S_basis}
Every tropical $\mathfrak{S}$-Gröbner basis contains a finite 
tropical $\mathfrak{S}$-Gröbner basis.
\end{prop}
\begin{proof}
We refer to the proof of Proposition 14 of \cite{Arri-Perry}.
It uses an adapted Dickson's Lemma and since it is mostly a question
of monomial ideals, the transposition to the tropical setting is direct.
\end{proof}

\subsection{Linear algebra and existence}

For $x^\alpha \in T$ and $1 \leq i \leq n,$ 
let us denote by $Mac_{\leq_{sign x^\alpha e_i }}(F)$
the Macaulay matrix of the polynomials $x^\beta f_j$ such 
that $x^\beta f_j \leq x^\alpha f_i,$
ordered increasingly for the order on the $x^\beta e_j$'s.
One can perform a tropical LUP algorithm on $Mac_{\leq d}(F)$
(see Algo. \ref{algo:trop LUP}) and obtain all the leading monomials
in $I_{\leq_{sign} x^\alpha e_i}$.
This can be (theoretically) performed for all $x^\alpha e_i$ to
obtain the existence of an $\mathfrak{S}$-GB of $I.$

\subsection{More on signatures}

%

We define $\Sigma$ to be the set of signatures.

Thanks to Proposition \ref{prop:F5elimination},
not all $x^\alpha e_i$ can be a signature:
\begin{rem}
If $x^\alpha e_i \in LM(TSyz(F))$ then $x^\alpha e_i \notin \Sigma.$
 \label{rem:constancy_and_signature}
\end{rem}

%

We provide two lemmata to understand the compatibility 
of $\Sigma$ with basic operations on polynomials.

\begin{lem}
If $f,g \in I$ are such that $S(f)=S(g)$ and $LM(f) \neq LM(g)$, then there exist $a,b \in k^*$ such that $S(af+bg)<S(f)$ and $af+bg \neq 0$. \label{lem:baisse_signature}
\end{lem}
If one takes the point of view of linear algebra, the proof is direct.

\begin{lem}
  \label{lem:signature-multiplication}
If $g \in I$ and $\tau \in T$
then $S(\tau g)\leq \tau S(g).$
If moreover $\tau S(g) \in \Sigma$, then 
$S(\tau g)=\tau S(g)$ and for all $\mu \in T$ such that $\mu$ divides $\tau$, $S(\mu g) = \mu S(g)$.
\end{lem}
\begin{proof}
The first part is direct.
For the second part, one can show that it is 
possible to write that
$\tau g = h +r$
for some $h \in I$ of signature $\tau S(g),$
irreducible, and $r \in I_{<_{sign}  \tau S(g)}$
and conclude that $S(\tau g)=\tau S(g).$

For the last statement, assume that there exists a $\mu \in T$ dividing $\tau$ such that $S(\mu g) < \mu S(g)$.
Then $S(\tau g) = S(\frac{\tau}{\mu}\mu g)\leq \frac{\tau}{\mu} S(\mu g) < \frac{\tau}{\mu}\mu S(g) = \tau S(g)$, which is a contradiction.
\end{proof}

\section{Buchberger-F5 criterion} 
\label{sec:BF5_criterion}

In this section, we explain a criterion, the Buchberger-F5 criterion, on
which we build our F5 algorithm to compute 
tropical $\mathfrak{S}$-Gröbner bases.
It is an analogue of the Buchberger criterion which includes 
the F5 elimination criterion.

We need a slightly different notion of $S$-pairs, called here normal pairs
and can then state the Buchberger-F5 criterion.

\begin{deftn}[Normal pair]
Given $g_1,g_2 \in I,$ let $Spol(g_1,g_2)=u_1 g_1-u_2 g_2$ be the $S$-polynomial of
$g_1$ and $g_2,$ with for $i\in \{1,2\}$, $u_i = \frac{lcm(LM(g_1),LM(g_2))}{LT(g_i)}.$ 
We say that $(g_1,g_2)$
is a \textbf{normal pair} if:
\begin{enumerate}
\item the $g_i$'s are $\mathfrak{S}$-irreducible polynomials.
\item $S(u_i g_i)=LM(u_i)S(g_i)$ for $i=1,2.$
\item $S(u_1 g_1) \neq S(u_2 g_2).$
\end{enumerate} 
We define accordingly $d$-normal pairs in $I^{\leq d}.$\label{def:normal_pair}
\end{deftn}

\begin{theo}[Buchberger-F5 criterion]
Suppose that $G$ is a finite set of $\mathfrak{S}$-irreducible polynomials of $I=\left\langle f_1, \dots, f_s \right\rangle$ such that:
\begin{enumerate}
\item for all $\forall i \in \llbracket 1, s \rrbracket,$ there exists $g \in G$ such that $S(g)=e_i.$
\item for any $g_1,g_2 \in G$ such that $(g_1,g_2)$ is a normal pair, there exists
$g \in G$ and $t \in T$ such that $tg$ is $\mathfrak{S}$-irreducible and 
$tS(g)=S(tg)=S(Spol(g_1,g_2)).$
\end{enumerate}
Then $G$ is a $\mathfrak{S}$-Gröbner basis of $I.$ 
The analogue result using $d$-normal pairs to recognize an $\mathfrak{S}_d$-GB in $I^{\leq d}$ is also true.  \label{thm:F5crit}
\end{theo}
\begin{rem}
The converse of this result is clear.
\end{rem}

Theorem \ref{thm:F5crit} is an analogue of the Buchberger criterion for tropical $\mathfrak{S}$-Gröbner bases.
To prove it, we adapt the classical proof of the Buchberger criterion and the proof of the tropical Buchberger algorithm of Chan and Maclagan (Algorithm 2.9 of \cite{Chan:2013}).
We need two lemmata, the first one being elementary.


\begin{lem} \label{lem:factor_Spoly}
Let $x^\alpha, x^\beta, x^\gamma, x^\delta \in T$ and $P,Q \in A$
 be such that $LM(x^\alpha P)=LM(x^\beta Q)=x^\gamma$
 and $x^\delta=lcm(LM(P), LM(Q)).$
 Then \[Spol(x^\alpha P, x^\beta Q)=x^{\gamma - \delta} Spol(P,Q). \]
\end{lem}

\begin{lem} \label{lem:SGB_rewriting}
Let $G$ be an $\mathfrak{S}$-Gröbner basis of $I$ up to some signature $\sigma.$ Let $h \in I,$ be such that $S(h) \leq \sigma.$
Then there exist $r \in \mathbb{N},$ $g_1,\dots, g_r \in G,$
$Q_1,\dots, Q_r \in A$ such that for all $i$ and $x^\alpha$ a monomial of $Q_i,$
$S(x^\alpha g_i)=x^\alpha S(g_i) \leq S(h)$ and 
$LT(Q_i g_i) \leq LT(h),$
the $x^\alpha S(g_i)$'s are all distinct and non-zero, 
and, finally, we have
\[h=\sum_{i=1}^r Q_i g_i.\]
\end{lem}
\begin{proof}
  It is clear by linear algebra.
  One can form a Macaulay matrix whose rows correspond to polynomials $c \tau g$ with $\tau \in T, c \in k^*, g \in G$ such that $S(\tau g)=\tau S(g) \leq S(h)$.
  Only one row is possible per non-zero signature, and each monomial in $LM(I_{ \leq \sigma})$ is reached as leading term by only one row.
  It is then enough to stack $h$ at the bottom of this matrix and perform a tropical LUP form computation (see Algorithm \ref{algo:trop LUP}) to read the $Q_i$'s on the reduction of $h$.
\end{proof}
\begin{proof}[PROOF of Theorem \ref{thm:F5crit}]
We prove the main result by induction on the signature. 
We follow the order $\leq_{sign}$ for the induction.
It is clear for $\sigma =e_1$
and also for the fact we can pass from an $\mathfrak{S}$-GB
up to $<_{sign}e_i$ to $\leq_{sign} e_i.$
We write the elements of $G$ as $g_1,\dots,g_r$
for some $r \in \mathbb{Z}_{>0}.$

Let us assume that $G$ is an $\mathfrak{S}$-GB up to signature $<_{sign}\sigma$
for some signature $\sigma$ and let us prove it is a $\mathfrak{S}$-GB up to $\leq_{sign} \sigma.$
We can assume that all $g \in G$ satisfy $LC(g)=1.$
Let $P \in I$ be irreducible, with $LC(P)=1$ and such that $S(P)=\sigma.$
We prove that there is $\tau \in  T, g\in G$ such that
$LM(P)=LM(\tau g)$ and $S(\tau g)=\tau S(g)=\sigma.$

Our first assumption for $G$ implies that there exist at least one $g \in G$ and some $\tau \in T$ such that 
$\tau S(g)=S(P)=\sigma.$

If $LM(\tau g)=_{\leq} LM(P)$ we are done.
Otherwise, by Lemma \ref{lem:baisse_signature}, there exist some $a, b \in k^*$ such that $S(a P+b \tau g)=\sigma'$ for some $\sigma'<_{sign}\sigma$.

We can apply Lemma \ref{lem:SGB_rewriting} to $aP+ b \tau g$ and obtain that there exist $h_1,\dots,h_r \in A,$ such that $P=\sum_{i=1}^r h_i g_i,$ and for all $i,$ and $x^\gamma$ monomial of $h_i,$ the $x^\gamma S(g_i)=S(x^\gamma g_i) \leq_{sign} \sigma$ are all distincts.
We remark that $LT(P) \leq \max_i (LT(g_i h_i))$.
We denote by $m_i:=LT(g_i h_i).$

Among all such possible ways of writing $P$ as $\sum_{i=1}^r h_i g_i,$ we define $\beta$ as the \textbf{minimum} of the $\max_i (LT(g_i h_i))$'s.
Such a $\beta$ exists thanks to Lemma 2.10 in \cite{CM:2013} (adaptation to the non-homogeneous setting is for free).
We write $x^u=LM(\beta).$

If $LT(P)=_{\leq}\beta,$ then we are done. Indeed, there is then some $i$ and $\tau$
in the terms of $h_i$ such that $LM(\tau g_i)=LM(P)$
and $S(\tau g_i)=\tau S(g_i) \leq_{sign} \sigma.$

We now show that $LT(P)<\beta$ leads to a contradiction.

We can renumber the $g_i$'s so that:
\begin{itemize}
\item $\beta=_{\leq}m_1=_{\leq}\dots=_{\leq}m_{d}$.
\item $\beta > m_i$ for $i>d.$
\end{itemize}

We can assume that among the set of possible $(h_1,\dots,h_r)$ that reaches $\beta,$
we take one such that this $d$ is minimal.

Since $LT(P)<\beta,$ then we have $d \geq 2.$

We can write 
\begin{multline*}
  Spol(g_1,g_2)=LC(g_2) \frac{lcm(LM(g_1),LM(g_2))}{LM(g_1)}g_1\\
  -LC(g_1) \frac{lcm(LM(g_1),LM(g_2)}{LM(g_2)}g_2.
\end{multline*}


By construction, $LM(h_{1})S(g_{1}) \neq LM(h_{2})S(g_{2})$, so $(LM(h_{1})g_{1},LM(h_{2}g_{2})$ is a normal pair.
By Lemma~\ref{lem:factor_Spoly}, there exists a term $\mu$ such that $\mu\frac{lcm(LM(g_1),LM(g_2))}{LM(g_i)} = LM(h_{i})$ for $i \in \{1,2\}$.
So by Lemma~\ref{lem:signature-multiplication}, $(g_{1},g_{2})$ is a normal pair as well.

If $S(Spol(g_1,g_2))=\sigma,$
by the second property of the F5 criterion, we are done.

Otherwise, $S(Spol(g_1,g_2))<_{sign}\sigma.$
Moreover, 
let \[L= \frac{LC(h_1 g_1) }{LC(g_1)LC(g_2)}\frac{x^u}{lcm(LM(g_1),LM(g_2))}.\]
Then we have $S(L \cdot Spol(g_1,g_2)) \leq_{sign} \sigma$ thanks to Lemma~\ref{lem:factor_Spoly}.
Using the same construction as before with the first assumption of the F5 criterion and Lemmata~\ref{lem:baisse_signature} and~\ref{lem:SGB_rewriting}, 
we obtain some $h_i'$'s such that
 $L \cdot Spol(g_1,g_2)=\sum_{i=1}^r h_i' g_i,$
 $LT(h_i' g_i) \leq LT( L \cdot Spol(g_1,g_2))<\beta$ for all $i.$
 Furthermore, the signatures $S(x^{\alpha}g_{i}) = x^{\alpha}S(g_{i})$ for $i \in \{1,\dots,r\}$ and $x^{\alpha}$ in the support of $h'_{i}$ are all distincts. 
 
We then get:
\begin{flalign*}
 P=&\sum_{i=1}^r h_i g_i, \\
 =&  \sum_{i=1}^r h_i g_i- L \cdot Spol(g_1,g_2)+\sum_{i=1}^r h_i'g_i, \\
 =& \left(h_1-\frac{LC(h_1g_1)}{LC(g_1)}\frac{x^{u}}{LM(g_1)}+h_1' \right) g_1\\
&  +\left(h_2-\frac{LC(h_1g_1)}{LC(g_2)}\frac{x^{u}}{LM(g_{2})}+h_2' \right) g_2  +\sum_{i=3}^r \left( h_i + h_i' \right) g_i, \\
  =: & \sum_{i=1}^r \widetilde{h}_i g_i, \\
\end{flalign*}
where the $\widetilde{h}_{i}$'s are defined naturally.

By construction, $LT(\widetilde{h}_1 g_{1})< LT(h_1 g_{1}) = \beta$
and $LT(\widetilde{h}_i) \leq \beta$ for $i \leq d$
and $LT(\widetilde{h}_i) < \beta$ for $i > d.$

As a consequence, we have obtained a new expression for $f$
with either $max_i(LT(\widetilde{h}_i)) <\beta$ 
or this term attained stricly less than $d$ times, which is in either case
a contradiction with their definitions as minima.
So $LT(P)=_{\leq}\beta,$ which concludes the proof of the main result. 
It is then direct to adapt the previous proof to the case of an $\mathfrak{S}_d$-GB.\end{proof}

This theorem holds also for $\mathfrak{S}$-GB (or $\mathfrak{S}_d$-GB) up to a given signature.
We have the following variant as a corollary
for compatibility with sugar-degree:

\begin{prop}
Suppose that $d \in \mathbb{Z}_{>0},$ and
$G$ is a finite set of polynomials of $I$ such that:
\begin{enumerate}
\item Any $g \in G$ is $\mathfrak{S}_d$-irreducible in $I^{\leq d}.$ 
\item For all $ g_1, g_2 \in G$ we have $g_1, g_2$ and $Spol(g_1,g_2)$ in $I^{\leq d}.$ 
\item For all $i \in \llbracket 1, s \rrbracket,$ there exists $g \in G$ such that $S_d(g)=e_i.$
\item for any $g_1,g_2 \in G$ such that $(g_1,g_2)$ is a $d$-normal pair, there exists
$g \in G$ and $t \in T$ such that $tg$ is $\mathfrak{S}_d$-irreducible and 
$tS_d(g)=S_d(tg)=S_d(Spol(g_1,g_2)).$
\end{enumerate}
Then $G$ is an $\mathfrak{S}$-Gröbner basis of $I.$ \label{prop:F5crit_sugar_bounded}
\end{prop}

\section{F5 algorithm}
\label{sec:F5_algo}

In this section, we present our F5 algorithm. 
To this intent, we need to discuss some crucial algorithmic
points: how to recognize with which pairs to proceed and
how to build the Macaulay matrices and
reduce them. Some algorithms are on the following page.

\subsection{Admissible pairs and guessed signatures}

The second condition in the Definition \ref{def:normal_pair}
 of normal pairs is not possible to check in advance in
 an F5 algorithm.
One needs an $\mathfrak{S}$-Gröbner basis up to the corresponding signature
to be able to certify it.
To circumvent this issue, we use the weaker
notion of admissible pair.

\begin{deftn}[d-Admissible pair] \label{def:adm_pairs}
Given $g_1,g_2 \in I^{\leq d},$ let $Spol(g_1,g_2)=u_1 g_1-u_2 g_2$ be the $S$-polynomial of
$g_1$ and $g_2.$ We have \[u_i = \frac{lcm(LM(g_1),LM(g_2))}{LT(g_i)}.\] We say that $(g_1,g_2)$
is a $d$-\textbf{admissible pair} if:
\begin{enumerate}
\item $LM(u_i)S_d(g_i)=x^\alpha_i e_{j_i} \notin LM(TSyz).$
\item $LM(u_1)S_d( g_1) \neq LM(u_2)S_d( g_2).$
\end{enumerate}
\end{deftn}

To certify that a set is an $\mathfrak{S}_d$-GB, handling $d$-admissible pairs instead of $d$-normal pairs is harmless. Indeed, $d$-normal pairs
in $I^{\leq d}$ are contained inside the $d$-admissible pairs.
Whether a pair is $d$-admissible can be checked easily
before proceeding to reduction.


During the execution of the algorithm, when a polynomial
$x^\alpha g$ is processed, it is at first
not possible to know what is its signature.
Algorithm \ref{F5_algo_complete} has computed
$S_d(g)$ beforehand.
Thanks to the F5 elimination criterion (Prop \ref{prop:F5elimination}), we can
detect some of the $x^\alpha g$ such that $S(x^\alpha g) \neq x^\alpha S(g)$ and eliminate them.
For the processed polynomials, we use $x^\alpha S_d(g)$ as a \textbf{guessed signature}
in the algorithm. 
Once an $\mathfrak{S}$-GB up to signature $<x^\alpha S_d(g)$
is computed, we have the following alternative.
First case: $S_d(x^\alpha g)<x^\alpha S_d(g)$ and
$x^\alpha g$ reduces to zero (by the computed
 $\mathfrak{S}_d$-GB up to $d$-signature $<x^\alpha S_d(g)$). The guessed signature was wrong
but it is harmless as the polynomial is useless anyway.
Second case: $S_d(x^\alpha g)=x^\alpha S_d(g),$
and then the guessed signature is certified.
Once the criterion of Proposition \ref{prop:F5crit_sugar_bounded} is satisfied, all
signatures are certified.

What happens when we can obtain $f$ with signature
$S_d(f)=x^\alpha e_i$ in degree $d$, and $S_{d+1}(f)=x^\beta e_j <_{sign} x^\alpha e_i$ in degree $d+1$?
Thanks to the way Algorithm \ref{algo:SymbPreprocRewritten}
 handles polynomials, always looking for smallest
 signature available, $f$ and its multiples will then be 
 built using only the second way. The first way of writing
 will at most appear so as to be reduced by the second one.

\subsection{Symbolic Preprocessing and Rewritten criterion}

One of the main parts of the F5 algorithm \ref{F5_algo_complete}
is the Symbolic Preprocessing : Algorithm \ref{algo:SymbPreprocRewritten}.
From the current set of S-pairs,
sugar-degree $d$, and the current $\mathfrak{S}_{d-1}$-GB, it produces a Macaulay
matrix. One can read on the tropical reduction
of this matrix new polynomials to append
to the current basis to obtain an $\mathfrak{S}_d$-GB.
It mostly consists of detecting which pairs are admissible
and selecting a (complete) set of reductors.

A special part of the algorithm is the use 
of Rewritten techniques (due to Faugère (see \cite{F5})).

The idea is the following.
Once a polynomial has passed the F5 elimination criterion
and is set to appear in a Macaulay matrix,
it can be replaced by any
other multiple of an element of $G$ of the same $d$-signature.
Indeed, assuming correctness of the algorithm
without any rewriting technique, if one of them, $h,$ is of $d$-signature $x^\alpha e_i,$ the algorithm
computes a tropical $\mathfrak{S}$-Gröbner basis
 up to $d$-signature
$<_{sign} x^\alpha e_i.$ Hence, $h$ can be replaced by any other polynomial of same signature:
it would be reduced to the same polynomial.
By induction, one can prove that all of them can be replaced at the same time.
We also remark that this is still valid for replacing a row of a given
guessed $d$-signature by another of the same guessed $d$-signature.

One efficient way is to 
replace a polynomial $t \times g$ by the polynomial $x^\beta h$ ($h \in G)$ of same (guessed) $d$-signature $tS_d(g)$ such that $x^\beta$ has smallest degree.\footnote{Indeed,
such an $h$ can be considered as one of the most reduced possible.} Taking the sparsest available is another possibility. 
It actually leads to a substantial reduction of the running time of the F5 algorithm.

\IncMargin{1em}
\begin{algorithm}

\SetKwInOut{Input}{input}\SetKwInOut{Output}{output}

\Input{$P$, a set of $d-1$-admissible pairs of sugar-degree d and $G$ such that $G \cap I^{\leq d-1}$ is an $\mathfrak{S}_{d-1}$-GB}
\Output{A Macaulay matrix of degree d}
\For{$Q$ polynomial in $P$}
{
Replace $Q$ in $P$ by the polynomial $(uS(g),u \times g)$
with $g$ latest added to $G$ reaching the same guessed signature \;}
$C \leftarrow$ the set of the \textbf{monomials} of the polynomials in $P$ \;
$U \leftarrow $ the polynomials of $P$ with their signature, except only one polynomial is taken by guessed signature \;
$D \leftarrow \emptyset $ \;

\While{$C \neq D$}{
$m \leftarrow \max (C \setminus D)$ \;
$D \leftarrow D \cup \{ m \}$ \;
$V \leftarrow \emptyset$ \;
\For{$g \in G$}{
\If{$LM(g) \mid m$}{
$V \leftarrow V \cup \{(g, \frac{m}{LM(g)}) \}$ \;}
}
$(g, \delta) \leftarrow$ the element of $V$ with $\delta \times g$ of smallest guessed signature not already in the signatures of U, with tie-breaking by taking minimal $\delta$ (for degree then for $\leq_{sign}$) \;  
$U \leftarrow U \cup \{ \delta \times g \}$ \;
$C \leftarrow C \cup \{ \text{monomials of } \delta \times g \}$ \;
}
$M \leftarrow$ the polynomials of $U,$ written in a Macaulay matrix 
and ordered by increasing guessed signature \;			
\textbf{Return} $M$ \; 
 \caption{Symbolic-Preprocessing-Rewritten} \label{algo:SymbPreprocRewritten}
\end{algorithm}
\DecMargin{1em}


\subsection{Linear algebra}

To reduce the Macaulay matrices while respecting the signatures,
we use the following tropical LUP algorithm from \cite{Vaccon:2015}: Algorithm~\ref{algo:trop LUP}.
If the rows correspond to polynomials ordered by increasing signature,
it computes a row-reduction, respecting the signatures
with each non-zero row with a different leading
monomial.

\IncMargin{1em}
\begin{algorithm} 
\DontPrintSemicolon

 \SetKwInOut{Input}{input}\SetKwInOut{Output}{output}

 \Input{$M$, a Macaulay matrix of degree $d$ in $A$, with $n_{row}$ rows and $n_{col}$ columns, and $mon$ a list of monomials indexing the columns of $M.$}
 \Output{$\widetilde{M}$, the $U$ of the tropical LUP-form of $M$}

$\widetilde{M} \leftarrow M$ ;  \;
\eIf{$n_{col}=1$ or $n_{row}=0$ or $M$ has no non-zero entry}{
				Return $\widetilde{M}$  ;\;
	}{			
\For{$i=1$ to $n_{row}$}{
\textbf{Find} $j$ such that $\widetilde{M}_{i,j}$ has the greatest term $\widetilde{M}_{i,j} x^{mon_j}$ for $\leq$ of the row $i$; \;
\textbf{Swap} the columns $1$ and $j$ of $\widetilde{M}$, and the $1$ and $j$ entries of $mon$; \;
By \textbf{pivoting} with the first row, eliminates the coefficients of the other rows on the first column; \;
\textbf{Proceed recursively} on the submatrix $\widetilde{M}_{i \geq 2, j \geq 2}$; \;}
\textbf{Return} $\widetilde{M}$; \;}

 \caption{The tropical LUP algorithm} \label{algo:trop LUP}
\end{algorithm}
\DecMargin{1em}

\subsection{A Complete Algorithm}

We now provide with Algorithm \ref{F5_algo_complete}
a complete version of an F5 algorithm
wich uses
Buchberger-F5 criterion and all the techniques
introduced in this section.

\IncMargin{1em}
\begin{algorithm} 
\DontPrintSemicolon

 \SetKwInOut{Input}{input}\SetKwInOut{Output}{output}

 \Input{$f_1,\dots,f_s$ polynomials, ordered by degree}
 \Output{A tropical $\mathfrak{S}$-GB $G$ of $\left\langle f_1,\dots,f_s \right\rangle$}

$G \leftarrow \{ (e_i,f_i) \text{ for i in } \llbracket 1,s\rrbracket \}$ ;  \;
$B \leftarrow \{ \text{S-pairs of } G \}$ ; $d \leftarrow 1$ ; \;
\While{$B \neq \emptyset$}{
\If{there is $i$ s.t. $\vert f_i \vert = d$}
{Replace the occurence of $f_i$ in $G$ by its reduction modulo $G \cap \left\langle f_1,\dots,f_{i-1} \right\rangle;$}
$P$ \textbf{receives} the pop of the $d-1$-admissible pairs in $B$ of sugar-degree $d$. Suppress from $B$ the others of sugar-degree $d$; \;
\textbf{Write} them in a Macaulay matrix $M_d$, along with their $\mathfrak{S}_d$-reductors obtained from $G$ (one per signature) by \textbf{Symbolic-Preprocessing-Rewritten}$(P,G)$ (Algorithm \ref{algo:SymbPreprocRewritten}); \;
\textbf{Apply} Algorithm \ref{algo:trop LUP} to compute the $U$ in the tropical LUP form of $M$ (no choice of pivot) ; \;
\textbf{Add} to $G$ all the polynomials obtained from $\widetilde{M}$ that provide new leading monomial up to their $d$-signature ; \;
\textbf{Add} to $B$ the corresponding new $d$-admissible pairs ; \;
$d \leftarrow d+1$ ; \;
}			
\textbf{Return} $G$ ; \; 

 \caption{A complete F5 algorithm} \label{F5_algo_complete}
\end{algorithm}
\DecMargin{1em}


\begin{theo}
Algorithm \ref{F5_algo_complete} computes an $\mathfrak{S}$-GB of $I.$ It avoids trivial tame syzigies.
\end{theo}
\begin{proof}
It relies on Theorem \ref{thm:F5crit} and then Proposition \ref{prop:F5crit_sugar_bounded}.
The proof is by induction on the sugar-degree,
then $i$, then the $x^\alpha e_i.$
One first proves that at the end of the main \emph{while} loop any guessed signature is correct, or its
row has reduced to zero, and then
that $\mathfrak{S}_d$-GB are computed, signature by signature.
One can then apply \ref{prop:F5crit_sugar_bounded} on the output to conclude.
Termination is a consequence of
correctness and Prop.~\ref{prop:finiteness_S_basis}.
For the syzygies, it is a consequence of
Prop.~\ref{prop:F5elimination} and the fact that
trivial syzygies of leading monomial $x^\alpha e_i$ are such that
$x^\alpha \in LM(\left\langle f_1,\dots,f_{i-1} \right\rangle).$
\end{proof}

\begin{rem}
Condition 1 of \ref{thm:F5crit} and 3 of \ref{prop:F5crit_sugar_bounded}
is not satisfied when for some $i,$ $f_i \in \left\langle f_1,\dots,f_{i-1}\right\rangle.$ This is harmless as: \begin{enumerate*}
\item As soon as it is found by computation, no signature in $e_i$ will appear anymore.
\item The Buchberger-F5 criterion can be applied omitting $f_i$.
\end{enumerate*}
\end{rem}

\section{Other algorithms}
\label{sec:Other_algo}

\subsection{Iterative F5}

In this subsection, we present briefly another way of extending 
the F5 algorithm to the affine setting: a completely iterative way in the initial polynomials.
The idea is to compute tropical Gröbner bases for 
$\left\langle f_1 \right\rangle, \left\langle f_1, f_2 \right\rangle, \dots, \left\langle f_1, \dots ,f_s \right\rangle.$

This corresponds to using the position over term ordering on the signatures, or in terms of filtration, to the following filtration on $A^s$:

\begin{deftn}
We write that $x^\alpha e_i \leq_{incr} x^\beta e_j$ if:
\begin{enumerate}
\item if $i < j.$
\item if $i = j$  and 
 $\vert x^\alpha f_i \vert < \vert x^\beta f_j \vert.$
\item if $i = j$  and 
 $\vert x^\alpha f_i \vert = \vert x^\beta f_j \vert,$
and 
\begin{itemize}
\item $x^\alpha \notin LM(I_{i-1})$
and $x^\beta \in LM(I_{i-1}),$ or
\item both $x^\alpha, x^\beta \in LM(I_{i-1})$ and $x^\alpha \leq x^\beta,$ or
\item both  $x^\alpha, x^\beta \notin LM(I_{i-1})$ and $x^\alpha \leq x^\beta.$ 
\end{itemize}
\end{enumerate} \label{deftn:order_for_the_iterative_filtration}
\end{deftn}

\begin{prop}[\cite{F5}]
If $x^\alpha \in LM(I_{ i-1}),$
then the filtration is constant
at \[I_{\leq  x^\alpha e_i}. \] \label{prop:F5elimination_iterative}
\end{prop} 
\begin{proof}
We can write $x^\alpha + g = \sum_{j<i} a_j f_j,$
with for all $j$ $a_j \in I$, and $g \in I$
with no monomial in $LM(I_{i-1}).$
Then: $x^\alpha f_i = (-g)f_i + \sum_{j=1}^{i-1} (a_j f_i) f_j,$
and the filtration is constant at $I_{\leq  x^\alpha e_i}.$
\end{proof}

It is then possible to state a Buchberger-F5
criterion and provide an adapted F5 algorithm.
The two algorithms will then differ in the following way.
\textbf{1.} For a given $x^\alpha$ and $e_i,$
the vector space
$I_{<x^\alpha e_i}$ is much bigger in the iterative setting,
often of infinite dimension.
Thus, polynomials of signature $x^\alpha e_i$
can be more deeply reduced.
\textbf{2.} More syzygies can be avoided in the iterative setting.
\textbf{3.} However, many more matrices are to be produced: one for each $i$ and each necessary degree.
Construction of the matrices is not mutualised by degree anymore. 

\subsection{F4}

Another way to compute tropical Gröbner
bases for affine polynomials is to adapt 
Faugère's F4 algorithm \cite{F99}

Roughly, the F4 algorithm is an adaptation
of Buchberger's algorithm such that:
all S-polynomials of a given degree
are processed and reduced together in a big Macaulay
matrix, along with their reducers.
The algorithm carries on the computation
until there is no S-polynomials to reduce.

In a tropical setting, we have adapted the 
so-called "normal strategy" of F4 using 
the tropical LUP algorithm to reduce the 
Macaulay matrices.
We have used Algorithm \ref{algo:trop LUP} to reduce
the Macaulay matrices. So-called tropical row-echelon forms
(Algorithm 3.2.2 and 3.7.3 of \cite{Vaccon:2015}) are also possible,
enabling a trade-off between speed, thoroughness of the reduction
and loss in precision.

\section{Numerical experiments}
\label{sec:Num_exp}

A toy implementation of our algorithms 
in Sagemath \cite{Sage} 
is available
on \url{https://gist.github.com/TristanVaccon}.
We have gathered some numerical results in the following arrays. Timings are in seconds of CPU time.\footnote{Everything was performed on a Ubuntu 16.04
with 2 processors of 2.6GHz and 16 GB of RAM.}

\subsection{Benchmarks}

Here, the base field is $\mathbb{Q}$ with $2$-adic valuation.
We have applied 
the tropical F5 algorithm, Algorithm \ref{F5_algo_complete},
an iterative tropical F5, and a tropical F4 algorithm on
the Katsura $n$ and Cyclic $n$ systems for 
varying $n.$ Dots mean no conclusion in decent time.

\hspace{-.5cm}
\begin{tabular}{|c|c|c|c|c|c|c|c|}
\hline 
w=[0,\dots,0] & Katsura 4 & 5 & 6 & 7 & Cyclic 4 & 5 & 6 \\ 
\hline 
Trop F5 & .16 & 1.2& 1371 & • &  0.4 & 21& •   \\
\hline 
Iterative trop F5 & 0.3 & 1.9 &1172 & •&0.4  &21 &  • \\ 
\hline 
Trop F4 &.5& 5& 30 &• & 1.7   & 112 &    • \\ 
\hline 
\end{tabular} 

\hspace{-.5cm}
\begin{tabular}{|c|c|c|c|c|c|c|c|}
\hline 
$w=[(-2)^{i-1}]$ & Katsura  4 & 5 & 6 & 7 & Cyclic 4 & 5 & 6 \\ 
\hline 
Trop F5 & 0.15  &  0.8&17  & • & 0.18 & 11 & • \\
\hline 
Iterative trop F5 &0.18 &  1.1&20  & • & 0.18 & 11 & • \\ 
\hline 
Trop F4 &  0.2  &  1.7&  15& • &1  & 65 & • \\ 
\hline 
\end{tabular}

\subsection{Trop. F5+FGLM}

For a given $p,$ we take three polynomials with random coefficients in $\mathbb{Z}_p$
(using the Haar measure)
in $\mathbb{Q}_p[x,y,z]$ of degree $2 \leq d_1 \leq d_2 \leq d_3 \leq 4.$
We first compute a tropical Gröbner basis
for the weight $w=[0,0,0]$\footnote{Efficiency of this
choice regarding to the loss in precision was studied in the extended version of  \cite{Vaccon:2015}} and the grevlex monomial
ordering, and then apply an FGLM algorithm
(tropical to classical as in Chapter 9 of 
\cite{Vaccon:these})
to obtain a lex GB.
For any given choice of $d_i$'s, we repeat the experiment 50 times.
Coefficients of the initial polynomials are all given at some high-enough
precision $O(p^N)$ for no precision issue
to appear.
We can not provide a certificate on the monomials 
of the output basis though.
Results are compiled in the following arrays.

Firstly, an array for timings given as couples: average of the timings for
the tropical F5 part and for the FGLM part, with $D=d_1+d_2+d_3-2,$ the Macaulay bound.
We add that for $p=2,3,$ there is often a huge 
standard deviation 
on the timings of the F5 part.

\begin{small}

\hspace{-.7cm}
\begin{tabular}{|l|c|c|c|c|c|c|c|c|c|c|c|c|}
\hline 
 & \multicolumn{2}{|c|}{ $D=4$} & \multicolumn{2}{|c|}{5} & \multicolumn{2}{|c|}{6}
 & \multicolumn{2}{|c|}{7} & \multicolumn{2}{|c|}{8} & \multicolumn{2}{|c|}{9}   \\ 
\hline 
$p=2$ &  .7&0.2 & 2.5&0.5 & 18&2.3 & 300&11 & 50&37 & 145&138 \\ 
\hline 
3 &  .8&.2  &  .9&.5  &  4&2  &  9&11  &  16&37 &  80&144  \\ 
\hline 
101 &  0.3&.2  &  .5&.5  &  1&2  &  3&10  &  4.6&37  &  11&150  \\ 
\hline
65519 &  .4&.2  & .6&.6 & 1.3&2.6  & 3.5&11 & 5&39  &  10&132  \\ 
\hline  
\end{tabular}

\end{small} 

Coefficients of the output tropical GB or classical GB
are known at individual precision $O(p^{N-m}).$
We compute the total mean and max on those $m$'s on the obtained GB.
Results are compiled in the following array as couples of mean and max.
The first array is for the F5 part and the second for the precision
on the final result.

\begin{small}
\hspace{-.5cm}
\begin{tabular}{|l|c|c|c|c|c|c|c|c|c|c|c|c|}
\hline 
 & \multicolumn{2}{|c|}{ $D=4$} & \multicolumn{2}{|c|}{5} & \multicolumn{2}{|c|}{6}
 & \multicolumn{2}{|c|}{7} & \multicolumn{2}{|c|}{8} & \multicolumn{2}{|c|}{9}  \\ 
\hline 
$p=2$ &  1.3&13 & 1.3&13 & 1.3&14 & 1.5&13 & 1.4&17 & 1.3&15 \\ 
\hline 
3 &  .6&6  &  .7&8  &  .7&7  &  .6&7  &  .6&7 &  .6&10 \\ 
\hline 
101 &  0&1  &  0&1  &  0&1  &  0&2  &  0&2  &  0&1   \\ 
\hline
65519 &  0&0  & 0&0 & 0&1  & 0&0 & 0&0  &  0&0   \\ 
\hline  
\end{tabular}

\hspace{-.5cm}
\begin{tabular}{|l|c|c|c|c|c|c|c|c|c|c|c|c|c|c|}
\hline 
& \multicolumn{2}{|c|}{ $D=4$} & \multicolumn{2}{|c|}{5} & \multicolumn{2}{|c|}{6}
 & \multicolumn{2}{|c|}{7} & \multicolumn{2}{|c|}{8} & \multicolumn{2}{|c|}{9}  \\ 
\hline 
$p=2$ &  8&71 &17&170 & 58&393 & 167&913 & 290&1600 & 570&3900 \\ 
\hline 
3 &  5&38  &  13&114  &  27&230  &  81&640  &  167&1600  &  430&3100  \\ 
\hline 
101 & .2&11  &  0&2  &  1.3&80  &  4&210  &  8&407  &  0&2   \\ 
\hline
65519 &  0&0  & 0&0 & 0&0  & 0&0 & 0&0  &  0&0  \\ 
\hline  
\end{tabular}
\end{small}
Most of the loss in precision appears in the FGLM part.
In comparison, the F5 part is quite stable, and hence,
our goal is achieved.

\begin{scriptsize}

\bibliographystyle{plain}

\end{scriptsize}

\end{document}